\documentclass{llncs}

\usepackage[dvipsnames]{xcolor}
\usepackage{microtype}
\usepackage{amsmath}
\usepackage{amssymb}
\usepackage{stmaryrd}
\usepackage{cite}
\usepackage{paralist}
\usepackage{mathpartir}
\usepackage{color}

\usepackage{listings}
\usepackage{listings-golang} 
\usepackage{color}

\newcommand{\onlyFinal}[1]{#1}
\newcommand{\onlyTR}[1]{#1}

\makeatletter
\newcommand{\xRightarrow}[2][]{\ext@arrow 0359\Rightarrowfill@{#1}{#2}}
\makeatother

\newcommand{\gopherlyzer}{\textup{gopherlyzer}}
\newcommand{\Gopherlyzer}{\textup{Gopherlyzer}}

\newcommand{\mathem}{\sf}

\newcommand{\WHILE}{\mbox{\mathem while}}
\newcommand{\DO}{\mbox{\mathem do}}
\newcommand{\SKIP}{\mbox{\mathem skip}}
\newcommand{\GO}{\mbox{\mathem go}}
\newcommand{\CLOSE}{\mbox{\mathem close}}
\newcommand{\SELECT}{\mbox{\mathem select}}

\newcommand{\ELSE}{\mbox{\mathem else}}
\newcommand{\IF}{\mbox{\mathem if}}

\newcommand{\THEN}{\mbox{\mathem then}}

\newcommand{\bi}{\begin{array}[t]{@{}l@{}}}
\newcommand{\ei}{\end{array}}
\newcommand{\ba}{\begin{array}}
\newcommand{\ea}{\end{array}}
\newcommand{\bda}{\[\ba}
\newcommand{\eda}{\ea\]}
\newcommand{\bp}{\begin{quote}\tt\begin{tabbing}}
\newcommand{\ep}{\end{tabbing}\end{quote}}

\def\ruleform#1{{\setlength{\fboxrule}{1pt}\fbox{\normalsize $#1$}}}

\newcommand{\myirule}[2]{{\renewcommand{\arraystretch}{1.2}\ba{c} #1
                      \\ \hline #2 \ea}}

\newcommand{\rlabel}[1]{\mbox{(#1)}}
\newcommand{\turns}{\, \vdash \,}

\newcommand{\match}[3]{#1 \xRightarrow{#2} #3}

\newcommand{\decU}[1]{#1 \downarrow}

\newcommand{\conc}{\cdot}
\newcommand{\concSPACE}{\! \cdot \!}

\newcommand{\bigconc}{\bullet}

\newcommand{\FSA}[1]{{\cal FSA}(#1)}

\newcommand{\deriv}[2]{d_{#2}(#1)} 

\newcommand{\pderiv}[2]{pd_{#2}(#1)} 

\newcommand{\concPart}[1]{{\cal C}(#1)}

\newcommand{\false}{\mathit{False}}
\newcommand{\true}{\mathit{True}}

\newcommand{\len}[1]{\mathit{len}(#1)}
\newcommand{\approxP}[2]{#1 \leadsto #2}
\newcommand{\approxPID}[3]{#1 \stackrel{#3}{\leadsto} #2}

\newcommand{\Left}{\mathit{L}}
\newcommand{\Right}{\mathit{R}}

\newcommand{\threadID}[1]{\sharp(#1)}

\newcommand{\semB}[3]{#1 \turns #2 \Downarrow #3}
\newcommand{\semC}[3]{#1 \turns #2 \Rightarrow #3}
\newcommand{\semP}[3]{#1 \xRightarrow{#2} #3}

\newcommand{\sndEvt}[1]{{#1!}}
\newcommand{\rcvEvt}[1]{{#1?}}

\newcommand{\sndEvtID}[2]{\sndEvt{#1}^{#2}}
\newcommand{\rcvEvtID}[2]{\rcvEvt{#1}^{#2}}

\newcommand{\snd}[2]{#1^{\mathsf{s}} \leftarrow #2} 
\newcommand{\rcv}[2]{#1 \leftarrow #2^{\mathbf{r}}} 

\newcommand{\forkEff}[1]{\mathit{Fork(#1)}}
\newcommand{\forkEffN}[2]{\mathit{Fork^{#1}(#2)}}

\newcommand{\chan}[2]{\mathit{Chan}(#1)^{#2}}
\newcommand{\SyncChan}{\mathit{Chan}}

\newcommand{\ms}[1]{{\bf MS: #1}}
\newcommand{\pt}[1]{{\bf PT: #1}}

\newcommand{\leftQ}[2]{#1 \backslash #2}

\newcommand\semle\leqq 
\newcommand\semeq\equiv              
\newcommand\syneq=         

\newcommand\DirectDescendant\sqsubset
\newcommand\Descendant\preceq
\newcommand\TrueDescendant\prec

\newcommand\Angle[1]{\langle#1\rangle}
\newcommand\Bag[1]{\{\!\!\{#1\}\!\!\}}
\newcommand\CalC{\mathcal{C}}
\newcommand\Config[2]{\Angle{#1, \Bag{#2}}}

\newcommand\Override{\lhd}

\newcommand\BeforeOrEqual\le

\title{Static Trace-Based Deadlock Analysis for Synchronous Mini-Go}

\author{Kai Stadtm{\"u}ller\inst1 \and Martin Sulzmann\inst1  \and Peter Thiemann\inst2}
\institute{
  Faculty of Computer Science and Business Information Systems \\
  Karlsruhe University of Applied Sciences \\
  Moltkestrasse 30, 76133 Karlsruhe, Germany\\
  \email{kai.stadtmueller@live.de} \\
  \email{martin.sulzmann@hs-karlsruhe.de}
  \and 
  Faculty of Engineering, University of Freiburg\\ Georges-K{\"o}hler-Allee
  079, 79110 Freiburg, Germany \\
  \email{thiemann@acm.org}
}

\begin{document}

\maketitle

\begin{abstract}
  We consider the problem of static deadlock detection for programs
  in the Go programming language which make use of synchronous channel
  communications.
  In our analysis, regular expressions extended with a fork operator 
  capture the communication behavior of a program.
  Starting from a simple criterion that characterizes traces of deadlock-free
  programs, we develop automata-based methods to check for
  deadlock-freedom.
  The approach is implemented and evaluated with a series of examples.
\end{abstract}

\section{Introduction}

The Go programming language~\cite{golang} attracts increasing
attention because it offers an elegant approach
to concurrent programming with message-passing
in the style of Communicating Sequential Processes (CSP)~\cite{Hoare:1978:CSP:359576.359585}.
Although message passing avoids many of the pitfalls of concurrent
programming with shared state (atomicity violations, order violations,
issues with locking, and so on), it still gives rise to problems like deadlock.
Hence, the goal of our work is the static detection of deadlocks in Go
programs which make use of (synchronous) message-passing using the
unbuffered version of Go's channels.

\subsection{Related work}
\label{sec:related-work}

Leaving aside data races, deadlocks constitute one of the core
problems in concurrent programming. However, most work on static
detection of deadlocks on the programming language level deals with
shared-memory concurrency. 

Boyapati and coworkers \cite{DBLP:conf/oopsla/BoyapatiLR02} define a
type-based analysis that relies on a partial order on locks and
guarantees that well-typed programs are free of data races and
deadlocks.
The approaches by Williams and coworkers
\cite{DBLP:conf/ecoop/WilliamsTE05} and Engler and Ashcraft
\cite{DBLP:conf/sosp/EnglerA03} detect cycles in a precomputed 
static lock-order graph to highlight potential deadlocks. In distributed
and database systems, most approaches are dynamic but also involve
cycle detection in wait-for graphs (e.g., \cite{Huang:1990:DDD:77606.77611}). In these approaches, the
main points of interest are the efficiency of the cycle detection
algorithms and the methods employed for the construction and
maintenance of the wait-for graph.

Mercouroff \cite{DBLP:conf/mfps/Mercouroff91} employs abstract
interpretation for an analysis of CSP programs using an abstract
domain that approximates the number of messages sent between processes.
Colby~\cite{DBLP:conf/pepm/Colby95} presents an analysis that uses control paths to identify threads that may be
created at the same point and constructs the communication topology of the
program. A more precise control-flow analysis was proposed by Martel and
Gengler~\cite{DBLP:conf/spin/MartelG00}. Similar to our approach, in
their work the accuracy of the analysis is enhanced by analyzing
finite automata to eliminate some impossible communication traces.

For message-passing programs, there are 
 elaborate algorithms that attempt accurate matching of
communications in process calculi (e.g., the work of Ladkin and Simon
\cite{Ladkin1995}). However, they consider messages between fixed
partners whereas we consider communication between multiple partners
on shared channels.

Further analysis of message passing in the context of
Concurrent ML (CML)~\cite{Reppy:1999:CPM:317040} is based on
effect systems that abstract programs into regular-expression-like
behaviors with the goal of detecting finiteness of communication
topologies \cite{DBLP:conf/popl/NielsonN94}. 
The deadlock detection analysis of Christakis and Sagonas
\cite{DBLP:conf/padl/ChristakisS11} also 
constructs a static graph and searches it for cycles. 
Specific to Go, the paper by Ng and Yoshida \cite{DBLP:conf/cc/NgY16}
translates Go programs into a core calculus with session types and
then attempts to synthesize a global choreography that subsumes all
session. A program is deemed deadlock-free if this synchronization succeeds
and satisfies some side condition.
Like our work, they consider a fixed number of processes and
synchronous communication.
Section~\ref{sec:evaluation}  contains a more detailed comparison with this work.

Kobayashi~\cite{DBLP:conf/concur/Kobayashi06} considers deadlock detection for the $\pi$-calculus~\cite{Milner:1999:CMS:329902}.
His type inference algorithm infers usage constraints
among receive and send operations.
In essence, the constraints represent a dependency graph where
the program is deadlock-free if there are no circular dependencies
among send and receive operations.
The constraints are solved by reduction to Petri net reachability
\cite{DBLP:journals/acta/Kobayashi05}. 
A more detailed comparison with Kobayashi's work is given
in  Section~\ref{sec:evaluation}.


\subsection{Contributions}
\label{sec:contribution}

Common to all prior work is their reliance on automata-/graph-based methods.
The novelty of our work lies in the use of a symbolic deadlock detection method
based on forkable behavior.

Forkable behaviors in the form of regular expressions extended with
fork and general recursion were introduced by Nielson and Nielson \cite{DBLP:conf/popl/NielsonN94}
to analyze the communication topology of CML (which is based
on ideas from CSP, just like Go).
In our own recent work \cite{DBLP:conf/lata/SulzmannT16}, we establish some
important semantic foundations for forkable behaviors such as a compositional trace-based semantics
and a symbolic Finite State Automata (FSA) construction method via Brzozowski-style derivatives \cite{321249}.
In this work, we apply these results to statically detect deadlocks in Go programs.

Specifically, we make the following contributions:

\begin{itemize}
\item We formalize Mini-Go, a fragment of the Go programming language
  which is restricted to synchronous message-passing (Section \ref{sec:mini-go}).

\item  We approximate the communication behavior of Mini-Go
      programs with forkable behaviors (Section~\ref{sec:approximation}).

\item We define a criterion for deadlock-freedom in terms of the traces resulting
  from forkable behaviors. We give a decidable check for deadlock-freedom for
  a large class of forkable behaviors by applying the FSA construction
  method developed in prior work \cite{DBLP:conf/lata/SulzmannT16}. We also consider 
   improvements to eliminate false positives
  (Section~\ref{sec:static-analysis}).
  
\item We evaluate our approach with examples
  and conduct a comparison with closely related work (Section~\ref{sec:evaluation}).
\end{itemize}


\onlyTR{The appendix contains further details such as proofs etc.}


\section{Highlights}
\label{sec:highlights}

\begin{lstlisting}[float={tp},captionpos={b},caption={Message passing in Go},label={lst:messsage-passing-go}]
func sel(x, y chan bool) {
        z := make(chan bool)
        go func() { z <- (<-x) }()
        go func() { z <- (<-y) }()
        <-z
}       
func main() {
        x := make(chan bool)
        y := make(chan bool)
        go func() { x <- true }()
        go func() { y <- false }()      
        sel(x,y)
        sel(x,y)        
}
\end{lstlisting}
Before we delve into deadlocks and deadlock detection, we first
illustrate the message passing concepts found in Go with the example
program in Listing~\ref{lst:messsage-passing-go}.
The \lstinline{main} function creates two synchronous channels
\texttt{x} and \texttt{y} that transport Boolean values.
Go supports (a limited form of) type inference and therefore no type annotations are required.
We create two threads using the \texttt{go \textit{exp}} statement.
It takes an expression \texttt{\textit{exp}} and executes it in a
newly spawned go-routine (a thread).
Each of these expressions calls an anonymous function that performs a
send operation on one of the channels.
In Go, we write \texttt{x <- true} to send
value \texttt{true} via channel \texttt{x}.
Then we call the function \lstinline{sel} twice.
This function creates another Boolean channel \texttt{z} locally and
starts two threads that ``copy'' a value from one of the argument
channels to \lstinline{z}. In Go, we write \texttt{<-x} to receive a value
via channel \texttt{x}. Thus, \texttt{z <- (<-x)} sends a value received via channel \texttt{x}
to channel \texttt{z}.

So, the purpose of \lstinline{sel} is to choose a value which can either
be received via channel \texttt{x} or channel \texttt{y}.
As each channel is supplied with a value, each of the two calls to  \lstinline{sel} might be
able to retrieve a value.
While there is a schedule such that the \lstinline{main} program runs
to completion, it is also possible that
execution of the second \lstinline{sel} call will get stuck.
Consider the case that in the first call to \lstinline{sel} both helper threads
get to execute the receive operations on \texttt{x} and \texttt{y} and
forward the values to channel \texttt{z}.
In this case, only one of the values will be picked up by the
\lstinline{<-z} and returned, but the local thread with the other value
will be blocked forever waiting for another read on \texttt{z}. 
In the second call to \lstinline{sel}, none of the local threads can
receive a value from \texttt{x} or \texttt{y}, hence there will be no
send operation on \texttt{z}, so that the final receive
\lstinline{<-z} remains blocked.

Our approach to detect such devious situations is to express the communication
behavior of a program in terms of forkable behaviors.
For the \lstinline{main} function in
Listing~\ref{lst:messsage-passing-go}, we obtain the following
forkable behavior
\bda{c}
\forkEff{\sndEvt{x}} \concSPACE \forkEff{\sndEvt{y}}
\concSPACE \forkEff{\rcvEvt{x} \concSPACE \sndEvt{z_1}} \concSPACE \forkEff{\rcvEvt{y} \concSPACE \sndEvt{z_1}} \concSPACE \rcvEvt{z_1}
\concSPACE \forkEff{\rcvEvt{x} \concSPACE \sndEvt{z_2}} \concSPACE \forkEff{\rcvEvt{y} \concSPACE \sndEvt{z_2}} \concSPACE \rcvEvt{z_2}
\eda
We abstract away the actual values sent and write $\sndEvt{x}$ to denote
sending a message to channel \texttt{x}
and $\rcvEvt{x}$ to denote reception via channel \texttt{x}.
$\forkEff{}$ indicates a forkable (concurrent) behavior which
corresponds to \texttt{go} statements in the program.
The concatenation  operator $\conc$ connects two forkable behaviors in
a sequence.
The function calls to \lstinline{sel} are inlined and the local
channels renamed to $z_1$ and $z_2$, respectively.

The execution schedules of \lstinline{main} can be described by a
matching relation for forkable behaviors 
where we symbolically rewrite expressions.
Formal details follow later. Here are some possible matching steps for our example.

\bda{ll}
 &
\forkEff{\sndEvt{x}} \conc \forkEff{\sndEvt{y}}
\conc \forkEff{\rcvEvt{x} \conc \sndEvt{z_1}} \conc \forkEff{\rcvEvt{y} \conc \sndEvt{z_1}} \conc \rcvEvt{z_1}
\conc
\\ & \forkEff{\rcvEvt{x} \conc \sndEvt{z_2}} \conc \forkEff{\rcvEvt{y} \conc \sndEvt{z_2}} \conc \rcvEvt{z_2}
\\
\xRightarrow{} &
\Bag{\underline{\sndEvt{x}},
     \sndEvt{y},
     \underline{\rcvEvt{x}} \conc \sndEvt{z_1},
     \rcvEvt{y} \conc \sndEvt{z_1}, 
     \rcvEvt{z_1}
\conc \forkEff{\rcvEvt{x} \conc \sndEvt{z_2}} \conc \forkEff{\rcvEvt{y} \conc \sndEvt{z_2}} \conc \rcvEvt{z_2}}
\\
\xRightarrow{\sndEvt{x} \conc \rcvEvt{x}} &
\Bag{ \sndEvt{y},
     \underline{\sndEvt{z_1}},
     \rcvEvt{y} \conc \sndEvt{z_1}, 
     \underline{\rcvEvt{z_1}}
     \conc \forkEff{\rcvEvt{x} \conc \sndEvt{z_2}} \conc \forkEff{\rcvEvt{y} \conc \sndEvt{z_2}} \conc \rcvEvt{z_2}}
\\
\xRightarrow{\sndEvt{z_1} \conc \rcvEvt{z_1}} &
\Bag{ \sndEvt{y},
     \rcvEvt{y} \conc \sndEvt{z_1}, 
     \forkEff{\rcvEvt{x} \conc \sndEvt{z_2}} \conc \forkEff{\rcvEvt{y} \conc \sndEvt{z_2}} \conc \rcvEvt{z_2}}
\\
\xRightarrow{} &
\Bag{ \sndEvt{y},
     \rcvEvt{y} \conc \sndEvt{z_1}, 
     \rcvEvt{x} \conc \sndEvt{z_2},
     \rcvEvt{y} \conc \sndEvt{z_2},
     \rcvEvt{z_2}}
\\
\xRightarrow{\sndEvt{y}\conc\rcvEvt{y}\conc \sndEvt{z_2}\conc\rcvEvt{z_2}} &
\Bag{\rcvEvt{y} \conc \sndEvt{z_1},
     \rcvEvt{x} \conc \sndEvt{z_2}}
\eda

We first break apart the expression into its concurrent parts indicated by the multiset notation $\Bag{\cdot}$.
Then, we perform two rendezvous (synchronization) steps where the partners involved are underlined.
In essence, the first call to \texttt{sel} picks up the value sent via channel \texttt{x}.
The last step where we combine two synchronization steps (and also omit underline) shows that
the second call to \texttt{sel} picks up the value sent via channel \texttt{y}.
Note that the main thread terminates but as for each call to \texttt{sel}
one of the helper threads is stuck our analysis reports a deadlock.

As mentioned above, another possible schedule is that the first call to \texttt{sel}
picks up both values sent via channels \texttt{x} and \texttt{y}.
In terms of the matching relation, we find the following
\bda{ll}
 &
\forkEff{\sndEvt{x}} \conc \forkEff{\sndEvt{y}}
\conc \forkEff{\rcvEvt{x} \conc \sndEvt{z_1}} \conc \forkEff{\rcvEvt{y} \conc \sndEvt{z_1}} \conc \rcvEvt{z_1}
\conc
\\ & \forkEff{\rcvEvt{x} \conc \sndEvt{z_2}} \conc \forkEff{\rcvEvt{y} \conc \sndEvt{z_2}} \conc \rcvEvt{z_2}
\\
\xRightarrow{\sndEvt{x}\conc \rcvEvt{x} \conc \sndEvt{y} \conc \rcvEvt{y} \conc \sndEvt{z_1} \conc \rcvEvt{z_1}}
&
\Bag{\sndEvt{z_1},\rcvEvt{x} \conc \sndEvt{z_2}, \rcvEvt{y} \conc \sndEvt{z_2}, \rcvEvt{z_2}
    }

\eda
As we can see, the second helper thread of the first call to \texttt{sel} is stuck,
both helper threads of the second call are stuck as well as the main thread.
In fact, this is the deadlock reported by our analysis as we attempt to find minimal deadlock examples.

The issue in the above example can be fixed by making use of \emph{selective} communication
to non-deterministically choose among multiple communications.
\begin{lstlisting}
func selFixed(x, y chan bool) {
  select {
    case z = <-x:
    case z = <-y:
  }
}
\end{lstlisting}
The \SELECT\ statement blocks until one of the cases applies.
If there are multiple \SELECT\ cases whose communication is enabled,
the Go run-time system `randomly' selects one of those and proceeds with
it. Based on a pseudo-random number the select cases are permuted
  and tried from top to bottom.
Thus, the deadlocking behavior observed above disappears
as each call to \lstinline{selFixed} picks up either a value sent
via channel \texttt{x} or channel \texttt{y} but it will never consume values from both channels.


\section{Mini-Go}
\label{sec:mini-go}

We formalize a simplified fragment of the Go programming language where
we only consider a finite set of pre-declared, synchronous channels.
For brevity, we also omit procedures and first-class channels and
only consider Boolean values.

\begin{definition}[Syntax]
  \bda{lcll}
  x,y, &\dots & &
  \mbox{Variables, Channel Names}
  \\
  s & ::= & v \mid \SyncChan & \mbox{Storables}
  \\
  v & ::= & \true \mid \false & \mbox{Values}
  \\
  \textit{vs} & ::= & [] \mid v : \textit{vs} & \mbox{Value Queues}
  \\
  b & ::= & v \mid x \mid b \&\& b \mid !b & \mbox{Expressions}
  \\
  e,f & ::= & \rcv{x}{y} \mid \snd{y}{b} & \mbox{Receive/Send}
  \\
  p,q & ::= & \SKIP \mid \IF\ b \ \THEN\ p \ \ELSE\ q \mid \WHILE\ b \ \DO\ p \mid p ; q & \mbox{Commands}
  \\ & \mid & 
  \SELECT\ [e_i \Rightarrow p_i]_{i\in I}  & \mbox{Communications}
  \\ & \mid & \GO\ p   & \mbox{Threads}
  \eda
\end{definition}

Variables are either bound to Boolean values or to the symbol $\SyncChan$
which denotes a  synchronous channel.
Like in Go, we use the `arrow' notation for the send and receive
operations on  channels.
We label the channel name to distinguish receive
from send operations.
That is, from $\rcv{x}{y}$ we conclude that $y$ is the channel
via which we receive a value bound to variable $x$.
From $\snd{y}{b}$ we conclude that $y$ is the channel to which some
Boolean value is sent.
Send and receive communications are shorthands
for unary selections: $e = \SELECT\ [e \Rightarrow \SKIP]$. 

The semantics of a Mini-Go program is defined with a small-step semantics.
The judgment
$\semP{\Config S {p_1, \dots, p_n}}{T}{\Config{S'}{p'_1,\dots, p'_m}}$
indicates that execution of program threads $p_i$
may evolve into threads $p'_j$ with trace $T$.
The notation $\Bag{p_1,...,p_n}$ represents a multi-set of
concurrently executing programs $p_1$, ..., $p_n$.
For simplicity, we assume that all threads share a global state $S$
and that distinct threads have only variables bound to channels in common.

Program trace $T$ records the communication behavior as a sequence of symbols where
symbol $\sndEvt{x}$ represents a send operation on channel $x$
and symbol $\rcvEvt{x}$ represents a receive operation on channel $x$.
As we assume synchronous communication,
each communication step  involves exactly two threads as formalized in
the judgment $\semP{\Config S {p, q}}{T}{\Config{S'}{p',q'}}$. 

The semantics of Boolean expressions is defined with a big-step semantics
judgment $\semB{S}{b}{v}$, where $S$ is the state in which expression
$b$ evaluates to value $v$.
For commands, the judgment $\semC{S}{p}{q}$ formalizes one (small-)
step that executes a single statement. Thus,
we are able to switch among different program threads after each statement.
Here are the details.

\begin{definition}[State] A state $S$ is either empty, a mapping, or an
  override a state with a new mapping:
  $
   S \ ::= \ () \mid (x \mapsto s) \mid S \Override (x \mapsto s)
  $
 \end{definition}
We write $S(x)$ to denote state lookup.
We assume that mappings in the right operand of the map override $\Override$ take
precedence. They overwrite any mappings in the left operand.
That is, $(x \mapsto \true) \Override (x \mapsto \false) = (x \mapsto \false)$.
We assume that for each channel $x$ the state contains
a mapping $x \mapsto \SyncChan$.

\begin{definition}[Expression Semantics $\semB{S}{b}{v}$]
  \begin{mathpar}
    \semB{S}{\true}{\true}

    \semB{S}{\false}{\false}
    \\
    \myirule{S(x) = v}
    {\semB{S}{x}{v}}
    
    \myirule{\semB{S}{b_1}{\false}}
    {\semB{S}{b_1 \&\& b_2}{\false}}

    \myirule{\semB{S}{b_1}{\true} \ \ \semB{S}{b_2}{v}}
    {\semB{S}{b_1 \&\& b_2}{v}}         

    \myirule{\semB{S}{b}{\false}}{\semB{S}{!b}{\true}}

    \myirule{\semB{S}{b}{\true}}{\semB{S}{!b}{\false}}
  \end{mathpar}
\end{definition}

\begin{definition}[Commands $\semC{S}{p}{q}$]
  \begin{mathpar}
  \rlabel{If-T} \ \myirule{\semB{S}{b}{\true}}
          {\semC{S}{\IF\ b \ \THEN\ p \ \ELSE\ q}{p}}
    
  \rlabel{If-F} \ \myirule{\semB{S}{b}{\false}}
          {\semC{S}{\IF\ b \ \THEN\ p \ \ELSE\ q}{q}}

  \rlabel{While-F} \ \myirule{\semB{S}{b}{\false}}
          {\semC{S}{\WHILE\ b \ \DO\ p}{\SKIP}}

  \rlabel{While-T} \ \myirule{\semB{S}{b}{\true}}
          {\semC{S}{\WHILE\ b \ \DO\ p}{p ; \WHILE\ b \ \DO\ p}}

          \rlabel{Skip} \ \semC{S}{\SKIP ; p}{p}

  \rlabel{Reduce} \ \myirule{\semC{S}{p}{p'}}
          {\semC{S}{p;q}{p';q}}

  \rlabel{Assoc} \ \semC{S}{(p_1 ; p_2) ; p_3}{p_1 ; (p_2 ; p_3)}
  \end{mathpar}

\end{definition}

\begin{definition}[Communication Traces]
 \label{def:communications}
  \bda{lcll}
  T & ::= & \epsilon & \text{empty trace} \\
  & \mid & \sndEvt{x} & \text{send event}\\
  & \mid& \rcvEvt{x} & \text{receive event}\\
  & \mid & T \conc T & \text{sequence/concatenation}
  \eda
\end{definition}    

As we will see, the traces obtained by running a program are of a particular `synchronous' shape.

\begin{definition}[Synchronous Traces]
\label{def:sync-trace}  
  We say $T$ is a \emph{synchronous} trace if $T$ is of the following more restricted form.
  \bda{rcl}
    T_s & ::= & \varepsilon \mid \alpha \conc \bar{\alpha} \mid T_s \conc T_s
  \eda
  where $\bar{\alpha}$ denotes the complement of $\alpha$ and is defined as follows:
  For any channel $y$, $\overline{\rcvEvt{y}} = \sndEvt{y}$ and $\overline{\sndEvt{y}} = \rcvEvt{y}$.
\end{definition}

We assume common equality laws for traces such as associativity of $\conc$
and $\epsilon$ acts as a neutral element. That is,  $\epsilon \conc T = T$.
Further, we consider the two synchronous traces
$\alpha_1 \conc \overline{\alpha_1} \conc ... \conc \alpha_n \conc \overline{\alpha_n}$
and $\overline{\alpha_1}\conc\alpha_1 \conc ... \conc \overline{\alpha_n} \conc \alpha_n$
to be equivalent.

\begin{definition}[Synchronous Communications
    $\semP{\Config S {p, q}}{T}{\Config{S'}{p',q'}}$]
  \label{def:sync-comm}
  \bda{c}
  \rlabel{Sync} \ \
  \myirule{ \mbox{for $k \in I \ \  l \in J$ where}
     \\ e_k = \rcv{x}{y} \ \ f_l = \snd{y}{b}
     \\ S_1(y) = \SyncChan \ \
        \semB{S_1}{b}{v}  \ \ S_2 = S_1 \Override (x\mapsto v)}
          {\semP{\Config{S_1}{ \SELECT\ [e_i \Rightarrow p_i]_{i\in I} , \SELECT\ [f_j \Rightarrow q_j]_{j\in J} }}{\sndEvt{y}\conc\rcvEvt{y}}{\Config{S_2}{p_k , q_l }}}            
  \eda
          
\end{definition}

A synchronous communication step non-deterministically
selects matching communication partners from two select statements.
The sent value $v$ is immediately bound to variable $x$ as we consider
unbuffered channels here.
Programs $p_k$ and $q_l$ represent continuations for the respective
matching cases.
The communication effect is recorded in the trace
$\sndEvt{y}\conc\rcvEvt{y}$, which arbitrarily  places the send
before the receive communication.
We just assume this order (as switching the order yields
an equivalent, synchronous trace) and use it consistently in our formal development.

In the upcoming definition, we make use of the following helper operation:
$
  p \fatsemi q =
  \begin{cases}
    p  & q = \SKIP\\
    p ; q & \text{otherwise.}
  \end{cases}
 $
  Thus, one rule can cover the two cases
  that a \GO\ statement is final in a sequence or followed by another statement.
  If the \GO\ statement is final, the pattern $\GO\ p \fatsemi q$
  implies that $q$ equals $\SKIP$.
  See upcoming rule \rlabel{Fork}.
  Similar cases arise in the synchronous communication step.
  See upcoming rule \rlabel{Comm}. 

\begin{definition}[Program Execution
    $\semP{\Config S {p_1, \dots, p_n}}{T}{\Config{S'}{p'_1,\dots, p'_m}}$]
  \label{def:program-execution}
  \bda{c}
  \rlabel{Comm} \ \
  \myirule{\semP{\Config S {p_1,p_2}}{T}{\Config{S'}{p_1',p_2'}}}
          {\semP{\Config S {p_1\fatsemi p_1'' , p_2\fatsemi p_2'' ,
                p_3,  \dots, p_n}}{T}{
              \Config{S'}{p_1'\fatsemi p_1'', p_2'\fatsemi p_2'', p_3,
                \dots, p_n}}}
  \\ \\
  \rlabel{Step} \ \
  \myirule{\semC{S}{p_1}{p_1'}}
          {\semP{\Config S {p_1,\dots, p_n}}{\varepsilon}{\Config S {p_1',\dots, p_n}}}          
  \\ \\          
  \rlabel{Fork} \ \
  \semP{\Config S {\GO\ p_1 \fatsemi q_1, p_2 ,\dots, p_n}}{\varepsilon}{
    \Config S {p_1, q_1, p_2, \dots, p_n}}
  \\ \\
    \rlabel{Stop} \ \
  \semP{\Config S {\SKIP, p_2 ,\dots, p_n}}{\varepsilon}{
    \Config S {p_2, \dots, p_n}}
  \\ \\         
  \rlabel{Closure} \ \
  \myirule{
    \semP{\Angle{ S, P}}{T}{\Angle{S', P'}} \ \
    \semP{\Angle{ S', P'}}{T'}{\Angle{S'', P''}}
  }{
    \semP{\Angle{ S, P}}{T \conc T'}{\Angle{S'', P''}}
  }
  \eda
\end{definition}
Rule \rlabel{Comm} performs a synchronous communication step
whereas rule \rlabel{Step} executes a single step in one of the threads.
Rule \rlabel{Fork} creates a new thread.
Rule \rlabel{Stop} removes threads that have terminated.
Rule \rlabel{Closure} executes multiple program steps. It uses $P$ to
stand for a multiset of commands.

We are interested in identifying \emph{stuck programs} as
characterized by the following definition.

\begin{definition}[Stuck Programs]
\label{def:stuck-programs}  
  Let $\CalC = \Config S {p_1, \dots, p_n}$ where $n>1$ be some configuration which results from
  executing some program $p$.
  We say that $p$ is \emph{stuck} w.r.t.~$\CalC$ if each $p_i$ starts with a
  select statement\footnote{Recall that primitive send/receive
    communications are expressed in terms of select.} 
  and no reduction rules are applicable on $\CalC$.
  We say that $p$ is stuck if there exists a configuration $\CalC$
  such that $p$ is stuck w.r.t.~$\CalC$.
\end{definition}

A stuck program indicates that \emph{all} threads are \emph{asleep}.
This is commonly referred to as a \emph{deadlock}.
In our upcoming formal results, we assume that for technical reasons there must be at least two such threads.
Hence, a `stuck' program consisting of a single thread,
e.g.~$\snd{x}{\true}; \rcv{y}{x}$,
is not covered by the above definition.
Our implementation deals with programs in which only
a single or some of the threads are stuck.

Our approach to detect deadlocks is
to (1) abstract the communication behavior of programs in terms of forkable behaviors,
and then (2) perform some behavioral analysis to uncover deadlocks.
The upcoming Section~\ref{sec:approximation} considers the abstraction.
The deadlock analysis is introduced in Section~\ref{sec:static-analysis}.


\section{Approximation via Forkable Behaviors}
\label{sec:approximation}

Forkable behaviors extend regular expressions with a fork operator
and thus allow for a straightforward and natural approximation
of the communication behavior of Mini-Go programs.

\begin{definition}[Forkable Behaviors \cite{DBLP:conf/lata/SulzmannT16}]
  The syntax of forkable behaviors (or behaviors for short) is defined as follows:
  \bda{lcl}
  r,s,t &::= & \phi \mid \varepsilon \mid  \alpha \mid r+s \mid r \conc s \mid r^* \mid \forkEff{r}
  \eda
  where $\alpha$ are symbols from a finite alphabet $\Sigma$.
\end{definition}

We find the common regular expression operators
for alternatives ($+$), concatenation ($\conc$),
repetition (${}^*$) and a new fork operator $\forkEff{}$.
We write $\phi$ to denote the empty language and
$\varepsilon$ to denote the empty word.

In our setting, symbols $\alpha$ are send/receive communications of
the form $\sndEvt{x}$ and $\rcvEvt{x}$, where $x$ is a channel name
(viz.\ Definition~\ref{def:communications}).
As we assume that there are only finitely many channels, we can guarantee
that the set of symbols $\Sigma$ is finite.

A program $p$ is mapped into a forkable behavior $r$
by making use of judgments $\approxP{p}{r}$.
The mapping rules are defined by structural induction
over the input $p$.
Looping constructs are mapped to Kleene star.
Conditional statements and select are mapped to alternatives
and a sequence of programs is mapped to some concatenated behaviors.


\begin{definition}[Approximation
  $\approxP{p}{r}$]\label{def:approximation}
  \begin{mathpar}
    \approxP{\SKIP}{\varepsilon}

  \myirule{\approxP{p}{r} \ \ \approxP{q}{s}}
          {\approxP{\IF\ b \ \THEN\ p \ \ELSE\ q}{r + s}}
    
  \myirule{\approxP{p}{r}}
          {\approxP{\WHILE\ b \ \DO\ p}{r^*}}

  \myirule{\approxP{p}{r} \ \ \approxP{q}{s}}
          {\approxP{p ; q}{r \conc s}}

  \approxP{\rcv{x}{y}}{\rcvEvt{y}}

  \approxP{\snd{y}{b}}{\sndEvt{y}}

  \myirule{\approxP{e_i}{r_i}
           \ \ \approxP{p_i}{s_i}
           \ \ \mbox{for $i\in I$}
          }
          {\approxP{\SELECT\ [e_i \Rightarrow p_i]_{i\in I}}{\sum_{i \in I} \ r_i \conc s_i}}

  \myirule{\approxP{p}{r}}
          {\approxP{\GO\ p}{\forkEff{r}}}
  \end{mathpar}
\end{definition}
What remains is  to verify that the communication behavior of $p$
is safely approximated by $r$.
That is, we need to show that all traces resulting from executing $p$
are also covered by $r$.

A similar result appears already in the Nielsons' work \cite{DBLP:conf/popl/NielsonN94}.
However, there are significant technical differences
as we establish connections between
the traces resulting from program execution
to the trace-based language semantics for forkable behaviors
introduced in our prior work \cite{DBLP:conf/lata/SulzmannT16}.

In that work \cite{DBLP:conf/lata/SulzmannT16}, we give a semantic description
of forkable behaviors in terms of a language denotation $L(r,K)$.
Compared to the standard definition, we find
an additional component $K$ which represents a set of traces.
Thus, we can elegantly describe the
meaning of an expression $\forkEff{r}$ as the shuffling
of the meaning of $r$ with the `continuation' $K$.
To represent Kleene star in the presence of continuation $K$,
we use a fixpoint operation $\mu F$ that denotes the  least fixpoint
of $F$ in the complete lattice formed by the powerset of
$\Sigma^*$. Here, $F$ must be a monotone function on this lattice,
which we prove in prior work.

\begin{definition}[Shuffling]
The (asynchronous) shuffle $v \| w \subseteq \Sigma^*$ is the set of all
interleavings of words $v,w \in\Sigma^*$. It is
defined inductively by 
\begin{align*}
  \varepsilon \| w &= \{ w \}  &
  v \| \varepsilon &= \{ v \} &
  xv \| yw & = \{x\} \conc (v \| yw) \cup \{y\} \conc (xv \| w)
\end{align*}
The shuffle operation is lifted to languages by $L \| M = \bigcup \{ v \| w \mid v \in L, w \in M \}$.
\end{definition}  

\begin{definition}[Forkable Expression Semantics]
  For a trace language $K\subseteq \Sigma^*$, 
  the semantics of a forkable expression is defined inductively by
  \begin{displaymath}
    \begin{array}[t]{rl}
    L (\phi, K) &= \emptyset \\
    L (\varepsilon, K) &= K \\
    L (x, K) &= \{x \conc w \mid w \in K \} 
    \end{array}
    \qquad
    \begin{array}[t]{rl}
    L (r+s, K) &= L (r, K) \cup L (s, K) \\
    L (r \conc s, K) &= L (r, L (s, K)) \\
    L (r^*, K) &= \mu\, \lambda X. L (r, X) \cup K \\
    L (\forkEff r, K) &= L (r) \| K
    \end{array}
  \end{displaymath}
  As a base case, we assume 
  $L(r) = L (r, \{  \varepsilon \})$.
\end{definition}

Next, we show that when executing some program $p$
under some trace $T$,
the resulting program state can be approximated
by the left quotient of $r$ w.r.t.~$T$
where $r$ is the approximation of the initial program $p$.
This result serves two purposes.
(1) All communication behaviors found in a program
can also be found in its approximation.
(2) As left quotients can be computed via Brzozowski's derivatives \cite{321249},
we can employ his FSA methods for static analysis.
We will discuss the first point in the following.
The second point is covered in the subsequent section.

If $L_1$ and $L_2$ are sets of traces,
we write $\leftQ {L_1} L_2$ to denote the left quotient
of $L_2$ with $L_1$ where 
$\leftQ {L_1} L_2 = \{ w \mid \exists v\in L_1. v \conc w \in L_2  \}$.
We write $\leftQ x L_1$
as a shorthand for $\leftQ { \{ x \} } L_1$.
For a word $w$ we give the following inductive definition:
$\leftQ \varepsilon L = L$ and $\leftQ {x\conc w} L = \leftQ w (\leftQ x L)$.

To connect approximations of resulting programs to left quotients,
we introduce some matching relations which operate on behaviors.
To obtain a match we effectively rewrite a behavior into (parts of) some left quotient.
Due to the fork operation, we may obtain a multiset of (concurrent) behaviors
written $\Bag{r_1, ..., r_n}$. We sometimes use $R$ as a short-hand for
$\Bag{r_1, ..., r_n}$.
As in the case of program execution (Definition~\ref{def:program-execution}),
we introduce a helper operation
$
r \bigconc s =
\begin{cases}
  r  & s = \varepsilon \\
  r \conc s & \text{otherwise}
\end{cases}
$
to cover cases where a fork expression is either the final expression,
or possibly followed by another expression.
We write $\match{\cdot}{}{\cdot}$ as a short-hand
for $\match{\cdot}{\varepsilon}{\cdot}$.
We also treat $r$ and $\Bag{r}$ as equal.

\begin{definition}[Matching Relation]
  \label{def:matching-relation}
  \begin{mathpar}
    \ruleform{\match{r}{T}{s}}
    \\
    \rlabel{L} \ \match{r+s}{}{r}

    \rlabel{R} \ \match{r+s}{}{s}
    \\
    \rlabel{K$_n$} \ \match{r^*}{}{r \conc r^*}

    \rlabel{K$_0$} \ \match{r^*}{}{\varepsilon}

    \rlabel{X} \ \match{\alpha \conc r}{\alpha}{r}
    \\
    \rlabel{A1} \ \match{\varepsilon \conc r}{}{r}
    
    \rlabel{A2} \ \myirule{\match{r}{}{s}}{\match{r \conc t}{}{s \conc      t}}

    \rlabel{A3} \ \match{(r \conc s) \conc t}{}{r \conc (s \conc t)} 
  \end{mathpar}
  
\bda{c}
\ruleform{\match{\Bag{r_1,\dots, r_m}}{T}{\Bag{s_1, \dots , s_n}}}
\\ \\
\rlabel{F} \ \match{\Bag{\forkEff{r}\bigconc s, r_1,\dots,
    r_n}}{\varepsilon}{\Bag{s,r, r_1, \dots, r_n}}
\ \ \ \
\rlabel{C}
\ \myirule{
  \match{R}{T}{R'} \ \
  \match{R'}{T'}{R''}
}{
  \match{R}{T\conc T'}{R''}
}
\\ \\
\rlabel{S1} \ \myirule{\match{r}{T}{s}}
        {\match{\Bag{r, r_1,\dots, r_n}}{T}{\Bag{s, r_1, \dots, r_n}}}
\ \ \ \
\rlabel{S2} \ \match{\Bag{\varepsilon, r_1, \dots, r_n}}{\varepsilon}{\Bag{r_1,\dots, r_n}}
\eda       
\end{definition}

We establish some basic results for the approximation and matching relation.
The following two results show that matches are indeed left quotients.

\begin{proposition}
  \label{prop:re-match-sound}
  Let $r$, $s$ be forkable behaviors and $T$ be a trace such that $\match{r}{T}{s}$.
  Then, we find that $L(s) \subseteq \leftQ T L(r)$.
\end{proposition}

\begin{proposition}
  \label{prop:re-match-soundness}
  Let $r_1$,...,$r_m$, $s_1$,...,$s_n$ be forkable behaviors and $T$ be a trace such that
  $\match{\Bag{r_1, \dots, r_m}}{T}{\Bag{s_1, \dots, s_n}}$.
  Then, we find that $L(s_1) \| ... \| L(s_n) \subseteq \leftQ T (L(r_1) \| ... \| L(r_m))$.
\end{proposition}

Finally, we establish that all traces resulting during program execution
can also be obtained by the match relation.
Furthermore, the resulting behaviors are approximations of the resulting programs.

\begin{proposition}
  \label{prop:sound-cmd-approx}
  If $\semC{S}{p}{q}$ and $\approxP{p}{r}$
  then $\match{r}{}{s}$ for some $s$ where $\approxP{q}{s}$.
\end{proposition}

\begin{proposition}
\label{prop:sound-approx}
  If $\semP{\Config S {p_1,...,p_m}}{T}{\Config {S'} {q_1,...,q_n}}$
  and $\approxP{p_i}{r_i}$ for $i=1,...,m$
  then $\match{\Bag{r_1,...,r_m}}{T}{\Bag{s_1,...,s_n}}$
  where   $\approxP{q_j}{s_j}$ for $j=1,..,n$.
\end{proposition}


\section{Static Analysis}
\label{sec:static-analysis}

Based on the results of the earlier section,
all analysis steps can be carried out on the forkable behavior
instead of the program text.
In this section, we first develop a `stuckness` criterion in terms of forkable behaviors to identify
programs with a potential deadlock.
Then, we consider how to statically check stuckness.

\subsection{Forkable Behavior Stuckness Criterion}

\begin{definition}[Stuck Behavior]
  We say that $r$ is \emph{stuck} if and only if
  there exists $\match{r}{T}{\varepsilon}$ for some
  non-synchronous trace $T$.
\end{definition}

Recall Definition~\ref{def:sync-trace} for a description of synchronous traces.

The following result shows that if the stuck condition does \emph{not} apply,
we can guarantee the absence of a deadlock.
That is, non-stuckness implies deadlock-freedom.

\begin{proposition}
 \label{prop:stuck} 
  Let $p$ be a stuck program and $r$ be a behavior
  such that  $\approxP{p}{r}$.
  Then, $r$ is stuck.
\end{proposition}

The above result does not apply to stuck programs consisting of a single thread.
For example, consider
$p=\snd{x}{\true}; \rcv{y}{x}$
and $r = \sndEvt{x} \conc \rcvEvt{x}$
where $\approxP{p}{r}$.
Program $p$ is obviously stuck, however, $r$ is not stuck
because any matching trace for $r$ is synchronous. For example,
$\match{r}{\sndEvt{x}\conc\rcvEvt{x}}{\varepsilon}$.
Hence, Definition \ref{def:stuck-programs} assumes
that execution of program $p$ leads to some state
where all threads are asleep, so that we can construct
a non-synchronous trace for the approximation of $p$.

Clearly, the synchronous trace $\sndEvt{x}\conc\rcvEvt{x}$ is not observable
under any program run of $p$. Therefore, we will remove
such non-observable, synchronous traces from consideration.
Before we consider such refinements of our stuckness criterion,
we develop static methods to check for stuckness.

\subsection{Static Checking of Stuckness}

To check for stuckness, we apply an automata-based method where
we first translate the forkable behavior into an equivalent finite state machine (FSA)
and then analyze the resulting FSA for stuckness.
The FSA construction method for forkable behaviors
follows the approach described in our prior work \cite{DBLP:conf/lata/SulzmannT16}
where we build a FSA based on Brzozowski's derivative construction method \cite{321249}.

We say that a forkable behavior $r$ is \emph{well-behaved} if there is no fork inside a Kleene star expression.
The restriction to well-behaved behaviors guarantees finiteness (i.e.,
termination) of the automaton construction.

\begin{proposition}[Well-Behaved Forkable FSA \cite{DBLP:conf/lata/SulzmannT16}]
  Let $r$ be a well-behaved behavior.
  Then, we can construct an $\FSA{r}$ where the alphabet coincides
  with the alphabet of $r$ and states can be connected to behaviors such that
  (1) $r$ is the initial state and (2) for each non-empty trace $T=\alpha_1 \conc ... \conc \alpha_n$
  we find a path
   $r=r_0 \stackrel{\alpha_1}{\rightarrow} r_1 ... r_{n-1} \stackrel{\alpha_n}{\rightarrow} r_n$ in $\FSA{r}$
  such that  $\leftQ T L(r) = L(r_n)$.
\end{proposition}

The kind of FSA obtained by our method  \cite{DBLP:conf/lata/SulzmannT16}  guarantees that
all matching derivations (Definition~\ref{def:matching-relation})
which yield a non-trivial trace can also be observed in the FSA.

\begin{proposition}[FSA covers Matching]
\label{prop:fsa-covers-matching}  
  Let $r$ be a well-behaved behavior
  such that $\match{r}{T}{\Bag{s_1,...,s_m}}$ for some non-empty trace $T=\alpha_1\conc ... \conc \alpha_n$.
  Then, there exists a path
  $r=r_0 \stackrel{\alpha_1}{\rightarrow} r_1 ... r_{n-1} \stackrel{\alpha_n}{\rightarrow} r_n$ in $\FSA{r}$
  such that $L(s_1) \| ... \| L(s_m) \subseteq L(r_n)$.
\end{proposition}

Based on above, we conclude that stuckness of a behavior
implies that the FSA is \emph{stuck} as well. That is, we encounter a non-synchronous path.

\begin{proposition}
  Let $r$ be a well-behaved behavior such that $r$ is stuck.
  Then, there exists a path
  $r=r_0 \stackrel{\alpha_1}{\rightarrow} r_1 ... r_{n-1} \stackrel{\alpha_n}{\rightarrow} r_n$ in $\FSA{r}$
  such that $L(r_i) \not = \{ \}$ for $i=1,...,n$ and
  $\alpha_1 \conc ... \conc \alpha_n$ is a non-synchronous trace.
\end{proposition}

\begin{proposition}
  \label{prop:minimial-non-sync-path}
  Let $r$ be a well-behaved behavior such that $\FSA{r}$ is stuck.
  Then, any non-synchronous path that exhibits stuckness
  can be reduced to a non-synchronous path where a state appears at most twice along that path
\end{proposition}

  Based on the above, it suffices to consider minimal paths.
  We obtain these paths as follow.
  We perform a breadth-first traversal of the $\FSA{r}$ starting
  with the initial state $r$ to build up all paths which satisfy the following
  criterion: (1) We must reach a final state, and (2) a state may appear at most twice along a path.
  It is clear that the set of all such paths is finite and their length is finite.
  If among these paths we find a non-synchronous path, then the $\FSA{r}$ is stuck.

\begin{proposition}
  Let $r$ be a well-behaved behavior.
  Then, it is decidable if the $\FSA{r}$ is stuck.
\end{proposition}

Based on the above, we obtain a simple and straightforward to implement
method for static checking of deadlocks in Mini-Go programs.
Any non-synchronous path indicates a potential deadlock and due to the symbolic
nature of our approach, erroneous paths can be traced back to the program text for debugging purposes.

\subsection{Eliminating False Positives}
\label{sec:elminiation}

Naive application of the criterion developed in the
previous section yields many false positives. 
In our setting, a false positive is a
non-synchronous path that is present in the automaton $\FSA{r}$,
but which cannot be observed in any program run of $p$.
This section introduces an optimization to eliminate many false
positives. This optimization is integrated in our implementation.

For example, consider the forkable behavior
$r=\forkEff{\sndEvt{x} \conc \sndEvt{y}} \conc \rcvEvt{x} \conc \rcvEvt{y}$
resulting from the program
$p=\GO \ ( \snd{x}{\true} ; \snd{y}{\false}  ) ; \rcv{z}{x} ; \rcv{z}{y}$.
Based on our FSA construction method, we discover
the non-synchronous path
$
r \xrightarrow{\sndEvt{x} \conc \sndEvt{y} \conc \rcvEvt{x} \conc \rcvEvt{y}} \varepsilon
$
where $\varepsilon$ denotes some accepting state.
However, just by looking at this simple program it is easy to see that there is no deadlock.
There are two threads and for each thread, each program statement synchronizes with
the program statement of the other thread at the respective position.
That is, $\semP{\Config \_ {p}}{\sndEvt{x}\conc \rcvEvt{x} \conc \sndEvt{y} \conc \rcvEvt{y}}{\Config \_ {}}$.

So, a possible criterion to `eliminate' a non-synchronous path from consideration
seems to be to check if there exists an alternative synchronous permutation of this path. There are two cases where we need to be careful:
(1) Conditional statements and (2) inter-thread synchronous paths.

\paragraph{Conditional statements}

Let us consider the first case.
For example, consider the following variant of our example:
\bda{lcl}
r & = & \forkEff{\sndEvt{x} \conc \sndEvt{y}} \conc (\rcvEvt{x} \conc \rcvEvt{y} + \rcvEvt{y} \conc \rcvEvt{x})
\\
\\
p & = & \GO \ ( \snd{x}{\true} ; \snd{y}{\false}  ) ;
\\ && \IF\ \true \  \THEN\ (\rcv{z}{x} ; \rcv{z}{y}) \ \ELSE\ (\rcv{z}{y} ; \rcv{z}{x})
\eda
By examining the program text, we see that there is no deadlock as the
program will always choose the `if' branch.
As our (static) analysis conservatively assumes that both branches
may be taken,
we can only use a synchronous permutation to eliminate a non-synchronous path
if we do not apply any conditional statements along this path.
In terms of the matching relation from Definition \ref{def:matching-relation},
we can characterize the absence of conditional statements
if none of the rules \rlabel{L}, \rlabel{R}, \rlabel{K$_n$} and \rlabel{K$_0$}
has been applied.

\paragraph{Inter-thread synchronous paths}

The second case concerns synchronization within the same thread.
Consider yet another variant of our example:
\bda{c}
r=\forkEff{\sndEvt{x} \conc \rcvEvt{x}} \conc \sndEvt{y} \conc \rcvEvt{y}
\\
\\
p=\GO \ ( \snd{x}{\true} ; \rcv{z}{x}) ; \snd{y}{\false} ; \rcv{z}{y}
\eda
The above program will deadlock.
However, in terms of the abstraction, i.e.~forkable behavior, we find
that for the non-synchronous path there exists a synchronous permutation
which does not make use of any of the conditional matching rules,
e.g.~$\match{r}{\sndEvt{x}\conc\rcvEvt{x}\conc \sndEvt{y} \conc \rcvEvt{y}}{\Bag{}}$.
This is clearly not a valid alternative as for example $\sndEvt{x}$
and $\rcvEvt{x}$ result from the same thread.

To identify the second case, we assume that receive/send symbols $\alpha$
in a trace carry a distinct thread identifier (ID).
We can access the thread ID of each symbol $\alpha$ via
some operator $\threadID{\cdot}$.
Under our assumed restrictions (i.e., no forks inside of loops, which
is no $\GO$ inside a $\WHILE$ loop) it is straightforward
to obtain this information precisely.

We refine the approximation of a program's communication behavior
in terms of a forkable behavior such that communications carry additionally
the thread identification number.
  Recall that we exclude programs where there is a $\GO$ statement
  within a $\WHILE$ loop. Thus, the number of threads is statically known
  and thread IDs can be attached to communication symbols
  via a simple extension $\approxPID{p}{r}{i}$ of the relation $\approxP{p}{r}$.
  The additional component $i$ represents the identification number of the current thread.
  We start with $\approxPID{p}{r}{0}$ where $0$ represents the main thread.
 We write symbol $\sndEvtID{x}{i}$ to denote a transmission over channel $x$
which takes place in thread $i$.
Similarly, symbol $\rcvEvtID{x}{i}$ denotes reception over channel $x$ in thread $i$.
For each symbol, we can access the thread identification number via operator
$\threadID{\cdot}$ where $\threadID{\sndEvtID{x}{i}} = i$
and $\threadID{\rcvEvtID{x}{i}} = i$.

The necessary adjustments to Definition~\ref{def:approximation}
are as follows.

  \begin{mathpar}
    \approxP{\SKIP}{\varepsilon}{i}

  \myirule{\approxP{p}{r}{i} \ \ \approxP{q}{s}{i}}
          {\approxP{\IF\ b \ \THEN\ p \ \ELSE\ q}{r + s}{i}}
    
  \myirule{\approxP{p}{r}{i}}
          {\approxP{\WHILE\ b \ \DO\ p}{r^*}{i}}

  \myirule{\approxP{p}{r}{i} \ \ \approxP{q}{s}{i}}
          {\approxP{p ; q}{r \conc s}{i}}

  \approxP{\rcv{x}{y}}{\rcvEvtID{y}{i}}{i}

  \approxP{\snd{y}{b}}{\sndEvtID{y}{i}}{i}

  \myirule{\approxP{e_i}{r_i}{i}
           \ \ \approxP{p_i}{s_i}{i}
           \ \ \mbox{for $i\in I$}
          }
          {\approxP{\SELECT\ [e_i \Rightarrow p_i]_{i\in I}}{\sum_{i \in I} \ r_i \conc s_i}}{i}

  \myirule{\approxP{p}{r}{i+1}}
          {\approxP{\GO\ p}{\forkEff{r}}}{i}
  \end{mathpar}

We summarize our observations.

\begin{definition}[Concurrent Synchronous Permutation]
  Let $T_1$ and $T_2$ be two traces.
  We say that $T_1$ is a \emph{concurrent synchronous permutation} of $T_2$ iff
  (1) $T_1$ is a permutation of the symbols in $T_2$,
  (2) $T_1$ is a synchronous trace of the form
  $\alpha_1 \conc \overline{\alpha_1}\conc ... \conc \alpha_n \conc \overline{\alpha_n}$
  where $\threadID{\alpha_i} \not= \threadID{\overline{\alpha_i}}$
  for $i=1,...,n$.
\end{definition}

\begin{proposition}[Elimination via Concurrent Synchronous Permutation]
  Let $p$ be a program.
  Let $r$ be a well-behaved behavior such that $\approxP{p}{r}$.
  For any non-synchronous path $T$ in $\FSA{r}$, 
  there exists a synchronous path $T_1$, a non-synchronous path $T_2$
  and a concurrent synchronous permutation $T_3$ of $T_2$
  such that
    $\match{r}{T_1}{\Bag{r_1,...,r_m}}$,
  $\match{\Bag{r_1,...,r_m}}{T_2}{\Bag{}}$, and
  $\match{\Bag{r_1,...,r_m}}{T_3}{\Bag{}}$
  where in the last match derivation
  none of the rules \rlabel{L}, \rlabel{R}, \rlabel{K$_n$} and \rlabel{K$_0$}
  have been applied.
  Then, program $p$ is not stuck.
\end{proposition}

The `elimination' conditions in the above proposition
can be directly checked in terms of the $\FSA{r}$.
Transitions can be connected to matching rules.
This follows from the derivative-based FSA construction.
Hence, for each non-synchronous path in $\FSA{r}$ we can
check for a synchronous alternative. We simply consider
all (well-formed) concurrent synchronous permutations and
verify that there is a path which does not involve conditional transitions.

A further source for eliminating false positives is to distinguish
among nondeterminism resulting from selective communication
and nondeterminism due to conditional statements.
For example, the following programs yield
the same (slightly simplified) abstraction
\bda{c}
 r=\forkEff{\sndEvt{x}} \conc (\rcvEvt{x} + \rcvEvt{y})
 \\
 \\
 p_1 = \GO\ \snd{x}{\true} ; \SELECT\ [\rcv{z}{x} \Rightarrow \SKIP, \rcv{z}{y} \Rightarrow \SKIP]
 \\
 \\
  p_2 = \GO\ \snd{x}{\true} ; \IF \true \ \THEN\ \rcv{z}{x} \Rightarrow \ \ELSE\ \rcv{z}{y}
  \eda
  It is easy to see that there is a non-synchronous path,
  e.g.~$r \xrightarrow{\sndEvt{x} \conc \rcvEvt{y}} \varepsilon$.
  Hence, we indicate that the program from which this forkable behavior resulted may get stuck.
  In case of $p_1$ this represents a false positive because the non-synchronous path will
  not be selected.

The solution is to distinguish between both types of nondeterminism
by abstracting the behavior of $\SELECT$ via some new operator $\oplus$ instead of $+$.
We omit the straightforward extensions to Definition~\ref{def:approximation}.
In terms of the matching relation, $+$ and $\oplus$ behave the same.
The difference is that for $\oplus$ certain non-synchronous behavior
can be safely eliminated.

Briefly, suppose we encounter a non-synchronous path where the (non-synchronous) issue
can be reduced to
$\match{\Bag{\alpha_1 \oplus ... \oplus \alpha_n, \beta_1 \oplus ... \oplus \beta_m}}{\alpha_i \conc \beta_j}{\Bag{}}$
for some $i\in\{1,...,n\}$ and $j\in\{1,...,m\}$
where $\alpha_i \conc \beta_j$ is non-synchronous.
Suppose there exists $l\in\{1,...,n\}$ and $k\in\{1,...,m\}$
such that
$\match{\Bag{\alpha_1 \oplus ... \oplus \alpha_n, \beta_1 \oplus ... \oplus \beta_m}}{\alpha_l \conc \beta_k}{\Bag{}}$
and $\alpha_l \conc \beta_k$ is synchronous.
Then, we can eliminate this non-synchronous path.
The reason why this elimination step is safe is
due to rule \rlabel{Sync} in Definition \ref{def:sync-comm}.
This rule guarantees that we will always synchronize if possible.
As in case of the earlier `elimination' approach, we can directly check the $\FSA{r}$
by appropriately marking transitions due to $\oplus$.

Further note that to be a safe elimination method,
we only consider $\SELECT$ statements where case bodies are trivial, i.e.~equal $\SKIP$.
Hence, we find $\alpha_i$ and $\beta_j$ in the above instead of arbitrary behaviors.
Otherwise, this elimination step may not be safe.
For example, consider
\bda{c}
r = \forkEff{\sndEvt{x} \conc \sndEvt{y}} \conc (\rcvEvt{x} \conc \rcvEvt{y} \oplus \rcvEvt{x} \conc \rcvEvt{x})
\\
\\
p = \GO\ (\snd{x}{\true}; \snd{y}{\false}) ;  \SELECT\ [\rcv{z}{x} \Rightarrow \rcv{z}{y}, \rcv{z}{x} \Rightarrow \rcv{z}{x}]
\eda
Due to the non-trivial case body $\rcv{z}{x}$ we encounter a non-synchronous path which
cannot be eliminated.



\section{Experimental Results}
\label{sec:evaluation}

\subsection{Implementation}

We have built a prototype of a tool that implements our approach,
referred to as {\gopherlyzer} \cite{gopherlyzer}.
Our analysis operates on the Go source language where
we make use of the oracle tool \cite{golang-oracle} to obtain
(alias) information to identify matching channel names.
We currently assume that all channels are statically known 
and all functions can be inlined.
The implementation supports \SELECT\ with default cases, something which
we left out in the formal description for brevity. Each default case
is treated as an empty trace $\varepsilon$.

Go's API also contains a \CLOSE\ operation for channels.
Receiving from a closed channel returns a default value whereas
sending produces an error.
An integration of this feature in our current implementation
is not too difficult but left out for the time being.
\onlyFinal{The technical report provides further details.}

\Gopherlyzer\ generates the FSA `on-the-fly' while processing the program.
It stops immediately when encountering a deadlock.
We also aggressively apply the `elimination' methods
described in Section~\ref{sec:elminiation} to reduce the size of the FSA.
When encountering a deadlock, the tool reports a minimal trace to highlight
the issue. We can also identify stuck threads by checking
if a non-synchronous communication pattern arises for this thread.
Thus, we can identify situations
where the main thread terminates but some local thread is stuck.
The reported trace could also be used to replay the synchronization steps
that lead to the deadlock. We plan to integrate
extended debugging support in future versions of our tool.

\subsection{Examples}

For experimentation, we consider the examples \texttt{deadlock}, \texttt{fanin},
and \texttt{primesieve}  from Ng and Yoshida \cite{DBLP:conf/cc/NgY16}.
To make \texttt{primesieve} amenable to our tool, we moved the dynamic creation of channels
outside of the (bounded) for-loop.
Ng and Yoshida consider two further examples:
\texttt{fanin-alt} and \texttt{htcat}.
We omit \texttt{fanin-alt} because our current implementation does not support
closing of channels. To deal with \texttt{htcat} we need to extend our
frontend to support
certain syntactic cases.
In addition, we consider the examples \texttt{sel} and \texttt{selFixed}
from Section~\ref{sec:highlights}
as well as \texttt{philo} which is a simplified implementation
of the dining philosophers problem where we assume
that all forks are placed in the middle of the table.
As in the original version, each philosopher requires two forks for eating.
All examples can be found in the \gopherlyzer\ repository \cite{gopherlyzer}.

\subsection{Experimental results}

\subsubsection*{Comparison with dingo-hunter \cite{DBLP:conf/cc/NgY16}}

\begin{table}[ht]
\begin{center}
  \begin{tabular}{l|l|l|l|l|l|lr|lr}
    Example & LoC & Channels & Goroutines & Select & Deadlock &
                                                                \multicolumn{2}{|l}{dingo-hunter}
    & \multicolumn{2}{|l}{ \gopherlyzer} \\
    &&&&&&result&time&result&time\\
    \hline
    deadlock & 34 & 2 & 5 & 0 & true & true & 155 & true & 21 \\
    fanin   &  37 & 3 & 4 & 1 & false & false & 107 & false & 29 \\
    primesieve & 57 & 4 & 5 & 0 & true & true & 8000 & true & 34 \\
    philo      & 34 & 1 & 4 & 0 & true & true & 480 & true & 31 \\
    sel        & 25 & 4 & 4 & 0 & true & true & 860 & true & 24 \\
    selFixed   & 25 & 2 & 2 & 2 & false & false & 85 & false & 30
  \end{tabular}
\end{center}
\caption{Experimental results. All times are reported in ms}
\label{tab:experimental-results}    
\end{table}

For each tool we report analysis results and the overall time used
to carry out the analysis.
Table~\ref{tab:experimental-results}  summarizes our results
which were carried out on some commodity hardware
(Intel i7 3770 @ 3.6GHz, 16 GB RAM, Linux Mint 17.3).

Our timings for dingo-hunter are similar to the reported
results \cite{DBLP:conf/cc/NgY16}, but it takes significantly
longer to analyze our variant of \texttt{primesieve}, where
we have unrolled the loop. 
There is also significant difference between \texttt{sel} and \texttt{selFixed}
by an order of magnitude.
A closer inspection shows that the communicating finite state machines (CFSMs)
generated by dingo-hunter
can grow dramatically in size with the number of threads and channels used.

The analysis time for our tool is always significantly faster
(between 3x and 235x with a geometric mean of 17x).
Judging from the dingo-hunter paper, the tool requires several transformation
steps to carry out the analysis, which seems rather time consuming.
In contrast, our analysis requires a single pass over the forkable behavior
where we incrementally build up the FSA to search for non-synchronous paths.

Both tools report the same analysis results.
We yet need to conduct a more detailed investigation but
it seems that both approaches are roughly similar
in terms of expressive power. However, there are some
corner cases where our approach appears to be superior.

Consider the following (contrived) examples in Mini-Go notation:
$(\GO\ \snd{x}{\true}) ; \rcv{y}{x}$
and
$\rcv{y}{x} ; (\GO\ \snd{x}{\true})$.
Our tool reports that the first example is deadlock-free
but the second example may have a deadlock.
The second example is out of scope of the dingo-hunter because it
requires all go-routines to be created before any communication takes
place. Presently, dingo-hunter does not seem to check this restriction
because it reports the second example as deadlock-free.


Our approximation with forkable behaviors
imposes no such restrictions. The first example yields
$\forkEff{\sndEvt{x}} \conc \rcvEvt{x}$
whereas the second example yields
$\rcvEvt{x} \conc \forkEff{\sndEvt{x}}$.
Thus, our tool is able to detect
the deadlock in case of the second example.


\subsubsection*{Comparison with Kobayashi~\cite{DBLP:conf/concur/Kobayashi06}}

We conduct a comparison with the TyPiCal tool \cite{typical}
which implements Kobayashi's deadlock analysis~\cite{DBLP:conf/concur/Kobayashi06}. As the source language is based on the $\pi$-calculus,
we manually translated the Go examples to the syntax
supported by TyPiCal's Web Demo Interface available from Kobayashi's
homepage. The translated examples 
can be found in the \gopherlyzer\ repository~\cite{gopherlyzer}.

To the best of our knowledge, TyPiCal does \emph{not} support
a form of \SELECT ive communication.
Hence, we need to introduce some helper threads
which results in an overapproximation of the original Go
program's behavior and potentially introduces a deadlock.
Recall the discussion in Section~\ref{sec:highlights}.

For programs not making use of selective communication
(and closing of channels; another feature not supported by TyPiCal),
we obtain the same analysis results.
Analysis times seem comparable to our tool.
The exception being the \texttt{primesieve} example which
cannot be analyzed within the resource limits imposed by TyPiCal's Web
Demo Interface which we used in the experiments.
Like our tool, TyPiCal properly maintains
the order among threads. Recall
the example $\rcv{y}{x} ; (\GO\ \snd{x}{\true})$ from above.

\begin{figure}[tp]
\begin{center}
  \begin{tabular}{ll}
    \begin{minipage}{6cm}
      \begin{lstlisting}
  // Go program        
  x := make(chan int)
  y := make(chan int)
  go func() {
    x <- 42
    v1 := <-y       // P1
    x <- 43
    v2 := <-y }()
  v3 := <-x
  v4 := <-x         // P2
  v5 := <-x
  y <- 42
  \end{lstlisting}
      Analysis report:
$
\sndEvtID{x}{1} \conc \rcvEvtID{x}{2} \conc \underline{\rcvEvtID{y}{1} \conc \rcvEvtID{x}{2}} \ldots
$      
    \end{minipage}
    &
    \begin{minipage}{6cm}

\begin{verbatim}
/*** TyPiCal input ***/  
new x in
new y in 
  x!42.y?v1.x!43.y?v2
| x?v3.x?v4.x?v5.y!42
\end{verbatim}

\begin{verbatim}
/*** TyPiCal output ***/
new x in
new y in 
  x!!42.y?v1.x!!43.y?v2
| x??v3.x?v4.x?v5.y!!42
\end{verbatim}
    \end{minipage}

  \end{tabular}
\end{center}
 \caption{Analysis Report: Gopherlyzer versus TyPiCal}
\label{fig:gopherlyzer-typical-debug-info} 
\end{figure}  

Finally, \gopherlyzer\ reports the analysis result in a different way
than TyPiCal.
The left side of Figure~\ref{fig:gopherlyzer-typical-debug-info} contains
a simple Go program and the right side its translation to TyPiCal's
source language.
TyPiCal reports that the program is unsafe and might deadlock.
Annotations \texttt{?} and \texttt{!} denote potentially stuck
receive and send operations whereas
\texttt{??} and \texttt{!!} indicate that the operations
might succeed.
The trace-based analysis (on the left) yields a non-synchronous trace
from which we can easily pinpoint the position(s) in the program
which are likely to be responsible.
In the example, the underlined events are connected
to program locations \texttt{P1} and \texttt{P2}.


\section{Conclusion}
\label{sec:conclusion}

We have introduced a novel trace-based static deadlock detection
method and built a prototype tool to analyze Go programs.
Our experiments show that our approach yields good results and its
efficiency compares favorably with existing tools of similar scope.

In future work, we intend to lift some of the restrictions of the
current approach, for example, supporting programs with dynamically
generated goroutines.
Such an extension may result in a loss of decidability of our static analysis.
Hence, we consider mixing our static analysis with some dynamic methods.

\section*{Acknowledgments}

We thank the APLAS'16 reviewers for their constructive feedback.

\bibliography{main}

\newpage

\appendix

\emph{Proofs and further details concerning select with default and closing of channels}


\section{Proofs}
\label{sec:proofs}

\subsection{Proof of Proposition \ref{prop:re-match-soundness}}

\begin{proposition}
  Let $r_1$,...,$r_m$, $s_1$,...,$s_n$ be forkable behaviors and $T$ be a trace such that
  $\match{\Bag{r_1, \dots, r_m}}{T}{\Bag{s_1, \dots, s_n}}$.
  Then, we find that $L(s_1) \| ... \| L(s_n) \subseteq \leftQ T (L(r_1) \| ... \| L(r_m))$.
\end{proposition}
\begin{proof}
   By induction on the derivation. We consider some of the cases.
  
  \textbf{Case \rlabel{S1}:}
  \bda{c}
  \myirule{\match{r}{T}{s}}
          {\match{\Bag{r, r_1, \dots,  r_n}}{T}{\Bag{s, r_1 , \dots, r_n}}}
  \eda
  By Proposition \ref{prop:re-match-sound},  $L(s) \subseteq \leftQ T L(r)$.
  We exploit the following facts:
   $\leftQ \alpha L(\alpha \conc r) = L(r)$ and
  $\leftQ \alpha (L_1 \| L_2) = ((\leftQ \alpha L_1) \| L_2) \cup (L_1 \| (\leftQ \alpha L_2))$.
  The desired result follows immediately.

  \textbf{Case \rlabel{F}:} By assumption $\match{\Bag{\forkEff{r},
      r_1,\dots, r_n}}{\varepsilon}{\Bag{r, r_1, \dots, r_n}}$.
  We have that $L(\forkEff{r}) = L(r) \| \{ \varepsilon \} = L(r)$.
  Thus, the desired result follows immediately.
  \qed
\end{proof}

\subsection{Proof of Proposition \ref{prop:sound-cmd-approx}}

\begin{proposition}
  If $\semC{S}{p}{q}$ and $\approxP{p}{r}$
  then $\match{r}{}{s}$ for some $s$ where $\approxP{q}{s}$.
\end{proposition}
\begin{proof}
  By induction.

  \textbf{Case \rlabel{If-T}:}
  \bda{c}
   \myirule{\semB{S}{b}{\true}}
           {\semC{S}{\IF\ b \ \THEN\ p_1 \ \ELSE\ p_2}{p_1}}
  \eda
  By assumption $\approxP{\IF\ b \ \THEN\ p_1 \ \ELSE\ p_2}{r_1 + r_2}$ for some $r_1$ and $r_2$
  where $\approxP{p_1}{r_1}$ and $\approxP{p_2}{r_2}$.
  Via rule \rlabel{L} we find that $\match{r_1 + r_2}{}{r_1}$ and we are done.

  \textbf{Case \rlabel{If-F}:} Similar to the above.

  \textbf{Case \rlabel{While-F}:}
  \bda{c}
   \myirule{\semB{S}{b}{\false}}
          {\semC{S}{\WHILE\ b \ \DO\ p}{\SKIP}}
   \eda
   By assumption $\approxP{\WHILE\ b \ \DO\ p}{r^*}$ for some $r$.
   Via rule \rlabel{K$_0$} we find that $\match{r^*}{}{\varepsilon}$.
   By definition $\approxP{\SKIP}{\varepsilon}$. Thus, we are done.

   \textbf{Case \rlabel{While-T}:}
   \bda{c}
      \myirule{\semB{S}{b}{\true}}
          {\semC{S}{\WHILE\ b \ \DO\ p}{p ; \WHILE\ b \ \DO\ p}}
   \eda
   By assumption $\approxP{\WHILE\ b \ \DO\ p}{r^*}$ for some $r$ where $\approxP{p}{r}$.
   Via rule \rlabel{K$_n$} we find that $\match{r^*}{}{r \conc r^*}$.
   Based on our assumption we find that
   $\approxP{p ; \WHILE\ b \ \DO\ p}{r \conc r^*}$ and thus we are done.

   \textbf{Case \rlabel{Skip}:} $\semC{S}{\SKIP ; p}{p}$.
   By assumption $\approxP{\SKIP ; p}{\varepsilon \conc r}$.
   Via rule \rlabel{A1} we conclude that $\match{\varepsilon \conc r}{}{r}$
   and are done.

   \textbf{Case \rlabel{Reduce}:}
   \bda{c}
       \myirule{\semC{S}{p}{p'}}
          {\semC{S}{p;p''}{p';p''}}
    \eda
    By assumption $\approxP{p; p''}{r \conc r''}$ where $\approxP{p}{r}$
    and $\approxP{p''}{r''}$ for some $r$ and $r''$.
    By induction $\match{r}{}{r'}$ for some $r'$ where $\approxP{p'}{r'}$.
    Via rule \rlabel{A2} we obtain $\match{r \conc r''}{}{r' \conc r''}$
    and we are done again.

    \textbf{Case \rlabel{Assoc}:} $\semC{S}{(p_1 ; p_2) ; p_3}{p_1 ; (p_2 ; p_3)}$.
    Follows via rule \rlabel{A3}.
    \qed
\end{proof}  

\subsection{Proof of Proposition \ref{prop:sound-approx}}

\begin{proposition}
  If $\semP{\Config S {p_1,...,p_m}}{T}{\Config {S'} {q_1,...,q_n}}$
  and $\approxP{p_i}{r_i}$ for $i=1,...,m$
  then $\match{\Bag{r_1,...,r_m}}{T}{\Bag{s_1,...,s_n}}$
  where   $\approxP{q_j}{s_j}$ for $j=1,..,n$.
\end{proposition}
\begin{proof}
 By induction.
 
\textbf{Case \rlabel{Fork}:}
\bda{c}
\semP{\Config S {\GO\ p_1 \fatsemi q_1, p_2,  ... , p_n}}{\varepsilon}{\Config S {p_1, q_1, p_2, ..., p_n}}
\eda

We assume $q_1 \not = \SKIP$.
By assumption $\approxP{\GO\ p_1 \fatsemi q_1}{r}$ for some $r$.
Hence, $r = \forkEff{r_1}\conc s_1$ where $\approxP{p_1}{r_1}$ and $\approxP{q_1}{s_1}$.
We further assume $\approxP{p_i}{r_i}$ for $i=2...n$.
Then, we find via rule \rlabel{F2}
$\match{\Bag{\forkEff{r_1}\conc s_1, r_2, ... ,r_n}}{}{\Bag{r_1, s_1, r_2, ... ,r_n}}$
and we are done.
For $q_1 = \SKIP$ the reasoning is similar.

\textbf{Case \rlabel{Step}:}
\bda{c}
  \myirule{\semC{S}{p_1}{p_1'}}
          {\semP{\Config S {p_1,...,p_n}}{\varepsilon}{\Config S {p_1',...,p_n}}}
\eda
By Proposition~\ref{prop:sound-cmd-approx} we find that $\match{r_1}{}{s_1}$ where $\approxP{p_1'}{s_1}$.
Via rule \rlabel{S1} we can conclude that
$\match{\Bag{r_1,...,r_n}}{\varepsilon}{\Bag{s_1,r_2,...,r_n}}$ and we are done.

\textbf{Case \rlabel{Stop}}: Via rule \rlabel{S2}.

\textbf{Case \rlabel{Closure}:}
By induction and application of rule \rlabel{C}.

\qed
\end{proof}

\subsection{Proof of Proposition \ref{prop:stuck}}

We require some auxiliary statements which both
can be verified by some straightforward induction.

The language denotation obtained is never empty.
\begin{proposition}
\label{prop:prog-is-non-phi}  
  Let $p$ be a program and $r$ be a forkable behavior such that $\approxP{p}{r}$.
  Then, we find that $L(r) \not= \{ \}$.
\end{proposition}

A non-empty language can always be matched against some trace.

\begin{proposition}
\label{prop:non-phi-yields-trace}  
  Let $r$ be a forkable behavior such that $L(r) \not= \{ \}$.
  Then, we find that $\match{r}{T}{\varepsilon}$ for some trace $T$.
\end{proposition}

\begin{proposition}
  Let $p$ be a stuck program and $r$ be a forkable behavior
  such that  $\approxP{p}{r}$.
  Then, $r$ is stuck.
\end{proposition}
\begin{proof}
  By assumption we find $\semP{\Config \_ {p}}{T_s}{\Config \_ {p_1,...,p_n}}$
  where $n>1$ and each $p_i$ starts with a communication primitive or a select statement.
  For brevity, we ignore the state component which is abbreviated by $\_$.
  By construction, $T_s$ is a synchronous trace.

  By Proposition~\ref{prop:sound-approx} we find
  $\match{r}{T_s}{\Bag{r_1,...,r_n}}$ where $\approxP{p_i}{r_i}$ for $i=1...n$.
 
  By assumption none of the $p_i$ can be reduced further. Recall that $n>1$.
  Hence, we must be able to further reduce at least two of the $r_i$'s
  such that we obtain a non-synchronous trace.
  For example, $\match{r}{T_s \conc \alpha \conc \beta}{\Bag{r_1', r_2',r_3,...,r_n}}$
  where $\overline{\alpha} \not = \beta$.
  Based on Propositions~\ref{prop:prog-is-non-phi} and~\ref{prop:non-phi-yields-trace}
  we can argue that $\Bag{r_1', r_2',r_3,...,r_n}$ can be further reduced.
  Hence, $\match{r}{T_s \conc \alpha \conc \beta \conc T}{\varepsilon}$ for some $T$.
  The overall trace $T_s \conc \alpha \conc \beta \conc T$ is non-synchronous.
  Thus, we can conclude that $r$ is stuck.
  \qed
 \end{proof}

\subsection{Proof of Proposition \ref{prop:fsa-covers-matching}}

\begin{proposition}[FSA covers Matching]
  Let $r$ be a well-behaved behavior
  such that $\match{r}{T}{\Bag{s_1,...,s_m}}$ for some non-empty trace $T=\alpha_1\conc ... \conc \alpha_n$.
  Then, there exists a path
  $r=r_0 \stackrel{\alpha_1}{\rightarrow} r_1 ... r_{n-1} \stackrel{\alpha_n}{\rightarrow} r_n$ in $\FSA{r}$
  such that $L(s_1) \| ... \| L(s_m) \subseteq L(r_n)$.
\end{proposition}
\begin{proof}
  By Proposition \ref{prop:re-match-soundness}
  we have that $L(s_1) \| ... \| L(s_m) \subseteq \leftQ T L(r)$.
  By property (2) for $\FSA{r}$ we find there exists a
     $r=r_0 \stackrel{\alpha_1}{\rightarrow} r_1 ... r_{n-1} \stackrel{\alpha_n}{\rightarrow} r_n$ in $\FSA{r}$
  such that  $\leftQ T L(r) = L(r_n)$.
  From above, we derive that
  $L(s_1) \| ... \| L(s_m) \subseteq \leftQ T L(r) = \leftQ T L(r) = L(r_n)$
  and we are done.
\end{proof}

\subsection{Proof of Proposition \ref{prop:minimial-non-sync-path}}

\begin{proposition}
  Let $r$ be a well-behaved behavior such that $\FSA{r}$ is stuck.
  Then, any non-synchronous path that exhibits stuckness
  can be reduced to a non-synchronous path where a state appears at most twice along that path
\end{proposition}
\begin{proof}
  By assumption the $\FSA{r}$ is stuck. We need to verify that for each non-synchronous path
  there exists a non-synchronous, minimal path.
  By minimal we mean that a state appears at most twice along that path.
  
  W.l.o.g., we assume the following
  $$
  r \stackrel{w_1}{\rightarrow} s \stackrel{w_2}{\rightarrow} s \stackrel{w_3}{\rightarrow} s \stackrel{w_4}{\rightarrow} t
  $$
  where state $s$ is repeated more than twice. There are possible further repetitions
  within the subpath $s \stackrel{w_4}{\rightarrow} t$
  but not within the subpath $r \stackrel{w_1}{\rightarrow} s$.
  By assumption $w_1 \conc w_2 \conc w_3 \conc w_4$ is non-synchronous.
  To show that we can derive a minimal non-synchronous path, we distinguish among
  the following cases.

  Suppose $w_1$ and $w_4$ are synchronous. Hence, either $w_2$ or $w_3$ must be non-synchronous.
  Suppose $w_2$ is non-synchronous.
  Then we can `simplify` the above to the non-synchronous example
  $r \stackrel{w_1}{\rightarrow} s \stackrel{w_2}{\rightarrow} s \stackrel{w_4}{\rightarrow} t$.
  A similar reasoning applies if $w_3$ is non-synchronous.

  Suppose $w_1$ is synchronous and $w_4$ is non-synchronous.
  Immediately, we obtain a `simpler' non-synchronous example
  $r \stackrel{w_1}{\rightarrow} s \stackrel{w_4}{\rightarrow} t$.
  
  Suppose $w_1$ is non-synchronous. We consider among the following subcases.
  Suppose that $w_1 = \alpha$.
  Suppose that $w_1 \conc w_2 \conc w_4$ and $w_1 \conc w_3 \conc w_4$ are synchronous
  (otherwise we are immediately done).
  Suppose that $w_1 \conc w_4$ is synchronous. We will show that this leads to a contradiction.
  From our assumption, we derive that $w_4 = \overline{\alpha} \conc w_4'$.
  As we assume that $w_1 \conc w_2 \conc w_4$ and $w_1 \conc w_3 \conc w_4$ are synchronous,
  we can conclude that $w_2 = \overline{\alpha}\conc ... \conc \alpha$
  and $w_3 = \overline{\alpha}\conc ... \conc \alpha$.
  However, this implies that $w_1 \conc w_2 \conc w_3 \conc w_4$ is synchronous
  which is a contradiction.
  Hence, either $w_1 \conc w_2 \conc w_4$ or $w_1 \conc w_3 \conc w_4$ is non-synchronous
  and therefore the example can be further simplified.

  Suppose that $w_1$ contains more than two symbols, e.g.~$w_1 = w_1' \conc \alpha$.
  If $w_1'$ is a non-synchronous, we can immediately conclude that
  $r \stackrel{w_1}{\rightarrow} s \stackrel{w_4}{\rightarrow} t$ is a (more minimal) non-synchronous path.
  Let us assume that $w_1'$ is synchronous.
  Then, we can proceed like above (case $w_1 = \alpha$) to show that
  a more minimal, non-synchronous path exists.
  
  These are all cases. Note that $w_1 = \varepsilon$ is covered
  by the above (cases where $w_1$ is assumed to be synchronous).
  \qed
\end{proof}

\section{Select with default}
\label{sec:select-default}

We show how to support \texttt{select}
with a default case written $\SELECT\ [e_i \Rightarrow q_i \mid q]_{i\in I}$.
The default case will only be executed if none of the other cases apply.

\bda{c}
  \rlabel{DefaultStep} \ \
  \myirule{\not \exists i \in I, j \in \{2,...,n\} \
            \semP{\Config S {q_i, p_i}}{\_}{\Config \_ {q_i',p_i}} }
          {\semP{\Config S {\SELECT\ [e_i \Rightarrow q_i \mid q]_{i\in I}\fatsemi p_1'' , p_2, \dots, p_n}}{}{
              \Config{S}{q\fatsemi p_1'', p_2, \dots, p_n}}}

\eda          

In case of the approximation, we represent the default case via $\varepsilon$.

\bda{c}
  \myirule{\approxP{e_i}{r_i}
           \ \ \approxP{q_i}{s_i}
           \ \ \mbox{for $i\in I$}
           \ \ \approxP{q}{s}
          }
          {\approxP{\SELECT\ [e_i \Rightarrow q_i \mid q]_{i\in I}}{(\sum_{i \in I} \ r_i \conc s_i}) + \varepsilon \conc s}

\eda

To eliminate false positives in the presence of \SELECT\ with default
we assume that a non-synchronous path can be eliminated/resolved
by making use of
$\match{\Bag{\alpha_1 \oplus ... \oplus \alpha_n \oplus \varepsilon, r_2,...,r_n}}{}{\Bag{r_2,...,r_n}}$.

\section{Closing Channels}
\label{sec:closing-channels}

In Go, a channel can be closed which means that no more values
which will be sent to it.
Any receive operation invoked after a channel is closed
will succeed and yield the default value.
However, any send operation leads to a `panic' which we consider as unsafe.

In terms of our analysis framework, we can integrate this additional
language feature by simply removing any receive event from the trace
which occurs after a channel has been closed.
As we generate traces from the FSA and FSA states can be connected to program
points, it is straightforward to identify the position in the trace
after which all receive events (for that channel) shall be removed.

For example, consider 
\begin{lstlisting}
  x := make(chan int)
  go func() {
    x <- 1 }()
  <-x
  close(x)
  <-x
\end{lstlisting}

Our analysis reports the non-synchronous trace
$$
\sndEvtID{x}{2} \conc \rcvEvtID{x}{1} \conc \rcvEvtID{x}{1}
$$
  where $\rcvEvtID{x}{1}$ represents
  the strictly non-synchronous portion of the trace.
  By taking into account the feature of closing of channels,
  this portion can be eliminated.
  Hence, our analysis reports that the program is safe.

  Consider the following variant where the close operation
  is part of the `then' branch of a conditional statement.
  The actual condition is omitted for brevity.
  
\begin{lstlisting}
  x := make(chan int)
  go func() {
    x <- 1 }()
  <-x  
  if ... {
    close(x)
  }  
  <-x
\end{lstlisting}

The tricky bit here is that the channel will only be closed if
the if-condition applies.
We therefore use a slightly refined language
of forkable behaviors to carry out the approximation
of the program's communication behavior.

$$
\forkEff{\sndEvtID{x}{2}} \conc \rcvEvtID{x}{1} \conc
(\mathit{close(x)} + \varepsilon) \conc \rcvEvtID{x}{1}
$$
where the  new event $\mathit{close(x)}$ represents
closing of a channel.

In the resulting FSA, we find the trace
$$
\sndEvtID{x}{2} \conc \rcvEvtID{x}{1} \conc \mathit{close(x)} \conc \rcvEvtID{x}{1}
$$
As any receive following a close operation will be non-blocking,
the program is safe for this specific program run.

There is however another alternative path reported by our analysis
which is unsafe
$$
\sndEvtID{x}{2} \conc \rcvEvtID{x}{1} \conc \varepsilon \conc \rcvEvtID{x}{1}
$$
We include the redundant $\varepsilon$ to highlight that the (implicit)
`else' branch was chosen.

\end{document}